\documentclass[fleqn]{llncs} 
\usepackage{hyperref}
\usepackage{amsmath} 
\usepackage{hlpa} 

\title{Using Hoare Logic in a Process Algebra Setting}
\author{J.A. Bergstra \and C.A. Middelburg}
\institute{Informatics Institute, Faculty of Science, University of
           Amsterdam \\
           Science Park~904, 1098~XH Amsterdam, the Netherlands \\
           \email{J.A.Bergstra@uva.nl,C.A.Middelburg@uva.nl}}
  
\begin{document}
\maketitle

\begin{abstract}
This paper concerns the relation between process algebra and Hoare 
logic.
We investigate the question whether and how a Hoare logic can be used 
for reasoning about how data change in the course of a process when 
reasoning equationally about that process.
We introduce an extension of \ACP\ (Algebra of Communicating Processes) 
with features that are relevant to processes in which data are involved, 
present a Hoare logic for the processes considered in this process 
algebra, and discuss the use of this Hoare logic as a complement to pure 
equational reasoning with the equational axioms of the process algebra.
\begin{classcode}
D.1.3, D.1.4, D.2.4, F.1.2, F.3.1.
\end{classcode}
\end{abstract}

\section{Introduction}
\label{sect-intro}

\ACP\ (Algebra of Communicating Processes) and its extensions provide a 
setting for equational reasoning about processes of some kind.  
The processes about which reasoning is in demand are often processes in 
which data are involved.
It is quite common for such a process that the data that are involved 
change in the course of the process and that the process proceeds at 
certain stages in a way that depends on the changing data.
This means that reasoning about a process often involves reasoning about
how data change in the course of that process.
The question arises whether and how a Hoare logic can be used for the 
second kind of reasoning when reasoning equationally about a process.
After all, processes of the kind described above are reminiscent of the 
processes that arise from the execution of imperative programs.

This paper is concerned with the above-mentioned question.
We investigate it using an extension of \ACP~\cite{BK84b} with features 
that are relevant to processes in which data are involved and a Hoare 
logic of asserted processes based on this extension of \ACP.
The extension concerned is called \deACPei.
Its additional features include assignment actions to deal with data 
that change in the course of a process and guarded commands to deal with 
processes that proceed at certain stages in a way that depends on 
certain data.
In the Hoare logic concerned, an asserted process is a formula of the 
form $\assproc{\phi}{p}{\psi}$, where $p$ is a term of \deACPei\ that 
denotes a process and $\phi$ and $\psi$ are terms of \deACPei\ that 
denote \linebreak[2] conditions.

We define what it means that an asserted process is true in such a way 
that $\assproc{\phi}{p}{\psi}$ is true iff a set of equations that 
represents this judgment is derivable from the axioms of 
\mbox{\deACPei}. 
Such a definition is a prerequisite for an affirmative answer to the 
question whether and how a Hoare logic can be used for reasoning about
how data change in the course of a process when reasoning equationally 
about that process.
The set of equations that represents the judgment expresses that a 
certain equivalence relation holds between processes determined by the 
asserted process.
The equivalence relation concerned may be a useful equivalence relation 
when reasoning about processes in which data are involved.
However, it is not a congruence relation, i.e.\ it is not preserved by 
all contexts.
This complicates pure equational reasoning considerably.
The presented Hoare logic can be considered to be a means to get 
partially round the complications concerned.

This paper is organized as follows.
We begin with presenting \ACPei, an extension of \ACP\ with the empty 
process constant $\ep$ and the binary iteration operator $\iter$, and 
\deACPei, an extension of \ACPei\ with features that are relevant to 
processes in which data are involved 
(Sections~\ref{sect-ACPei} and~\ref{sect-deACPei}). 
We also present a structural operational semantics of \deACPei, define a 
notion of bisimulation equivalence based on this semantics, and show 
that the axioms of \deACPei\ are sound with respect to this bisimulation 
equivalence (Section~\ref{sect-sos-bisim}).
After that, we present a Hoare logic of asserted processes based on 
\deACPei, define what it means that an asserted process is true, and 
show that the axioms and rules of this Hoare logic are sound with 
respect to this meaning (Section~\ref{sect-HL}).
Following this, we go further into the connection of the presented Hoare 
logic with \deACPei\ by way of the equivalence relation referred to in 
the previous paragraph (Section~\ref{sect-using-HL}).
We also go into the use of the presented Hoare logic as a complement to 
pure equational reasoning with the axioms of \deACPei\ by means of 
examples (Section~\ref{sect-discussion}).
Finally, we discuss related work and make some concluding remarks 
(Sections~\ref{sect-related} and~\ref{sect-conclusions}).

\section{\ACP\ with the Empty Process and Iteration}
\label{sect-ACPei}

In this section, we present \ACPei, \ACP~\cite{BK84b} extended with 
the empty process constant $\ep$ as in~\cite[Section~4.4]{BW90} and the 
binary iteration operator ${} \iter {}$ as in~\cite{BBP94a}.
In \ACPei, it is assumed that a fixed but arbitrary finite set $\Act$ of 
\emph{basic actions}, with $\dead,\ep \not\in \Act$, and a fixed but 
arbitrary commutative and associative \emph{communication} function 
$\funct{\commf}
 {(\Act \union \set{\dead}) \x (\Act \union \set{\dead})}
 {(\Act \union \set{\dead})}$, 
such that $\commf(\dead,a) = \dead$ 
for all $a \in \Act \union \set{\dead}$, have been given.
Basic actions are taken as atomic processes.
The function $\commf$ is regarded to give the result of synchronously
performing any two basic actions for which this is possible, and to be 
$\dead$ otherwise.
Henceforth, we write $\Actd$ for $\Act \union \set{\dead}$.

The algebraic theory \ACPei\ has one sort: the sort $\Proc$ of 
\emph{processes}.
We make this sort explicit to anticipate the need for many-sortedness 
later on. 
The algebraic theory \ACPei\ has the following constants and operators 
to build terms of sort~$\Proc$:
\begin{itemize}
\item
the \emph{inaction} constant $\const{\dead}{\Proc}$;
\item
the \emph{empty process} constant $\const{\ep}{\Proc}$;
\item
for each $a \in \Act$, the \emph{basic action} constant 
$\const{a}{\Proc}$;
\item
the binary \emph{alternative composition} operator 
$\funct{\altc}{\Proc \x \Proc}{\Proc}$;
\item
the binary \emph{sequential composition} operator 
$\funct{\seqc}{\Proc \x \Proc}{\Proc}$;
\item
the binary \emph{iteration} operator 
$\funct{\iter}{\Proc \x \Proc}{\Proc}$;
\item
the binary \emph{parallel composition} operator 
$\funct{\parc}{\Proc \x \Proc}{\Proc}$;
\item
the binary \emph{left merge} operator 
$\funct{\leftm}{\Proc \x \Proc}{\Proc}$;
\item
the binary \emph{communication merge} operator 
$\funct{\commm}{\Proc \x \Proc}{\Proc}$;
\item
for each $H \subseteq \Act$, 
the unary \emph{encapsulation} operator 
$\funct{\encap{H}}{\Proc}{\Proc}$.
\end{itemize}
We assume that there is a countably infinite set of variables of sort 
$\Proc$, which contains $x$, $y$ and $z$.
Terms are built as usual.
We use infix notation for the binary operators.
The following precedence conventions are used to reduce the need for
parentheses: the operator ${} \seqc {}$ binds stronger than all other 
binary operators and the operator ${} \altc {}$ binds weaker than all 
other binary operators.

The constants and operators of \ACPei\ are the constants and operators 
of \ACPe~\cite[Section~4.4]{BW90} and additionally the iteration 
operator ${} \iter$.
Let $p$ and $q$ be closed \ACPei\ terms, $a \in \Act$, and
$H \subseteq \Act$.%
\footnote
{As usual, a term in which no variables occur is called a closed term.}
Then the constants and operators of \ACPei\ can be explained as follows:
\begin{itemize}
\item
the constant $\dead$ denotes the process that is not capable of doing 
anything, not even terminating successfully;
\item
the constant $\ep$ denotes the process that is only capable of 
terminating successfully;
\item
the constant $a$ denotes the process that is only capable of first 
performing action $a$ and next terminating successfully;
\item
a closed term of the form $p \altc q$ denotes the process that behaves 
either as the process denoted by $p$ or as the process denoted by $q$, 
but not both;
\item
a closed term of the form $p \seqc q$ denotes the process that first 
behaves as the process denoted by $p$ and on successful termination of 
that process next behaves as the process denoted by $q$;
\item
a closed term of the form $p \iter q$ denotes the process that behaves 
either as the process denoted by $q$ or as the process that first 
behaves as the process denoted by $p$ and on successful termination of 
that process next behaves as $p \iter q$ again;
\item
a closed term of the form $p \parc q$ denotes the process that behaves 
as the processes denoted by $p$ and $q$ taking place in parallel, by 
which we understand that, each time an action is performed, either a 
next action of one of the two processes is performed or a next action of 
the former process and a next action of the latter process are performed 
synchronously;
\item
a closed term of the form $p \leftm q$ denotes the process that behaves 
the same as the process denoted by $p \parc q$, except that it starts 
with performing an action of the process denoted by $p$;
\item
a closed term of the form $p \commm q$ denotes the process that behaves 
the same as the process denoted by $p \parc q$, except that it starts 
with performing an action of the process denoted by $p$ and an action of 
the process denoted by $q$ synchronously;
\item
a closed term of the form $\encap{H}(p)$ denotes the process that 
behaves the same as the process denoted by $p$, except that actions from 
$H$ are blocked.
\end{itemize}

The axioms of \ACPei\ are the equations given in 
Table~\ref{axioms-ACPei}.
\begin{table}[!t]
\caption{Axioms of \ACPei}
\label{axioms-ACPei}
\begin{eqntbl}
\begin{axcol}
x \altc y = y \altc x                                  & \axiom{A1} \\
(x \altc y) \altc z = x \altc (y \altc z)              & \axiom{A2} \\
x \altc x = x                                          & \axiom{A3} \\
(x \altc y) \seqc z = x \seqc z \altc y \seqc z        & \axiom{A4} \\
(x \seqc y) \seqc z = x \seqc (y \seqc z)              & \axiom{A5} \\
x \altc \dead = x                                      & \axiom{A6} \\
\dead \seqc x = \dead                                  & \axiom{A7} \\
x \seqc \ep = x                                        & \axiom{A8} \\
\ep \seqc x = x                                        & \axiom{A9} \\
{}                                                                  \\
{}                                                                  \\
x \iter y = x \seqc (x \iter y) \altc y                & \axiom{BKS1} \\
z = x \seqc z \altc y \;\Limpl\; z = x \iter y         & \axiom{RSP*}
\end{axcol}
\quad
\begin{axcol}
x \parc y = x \leftm y \altc y \leftm x \altc x \commm y \altc  
                 \encap{\Act}(x) \seqc \encap{\Act}(y) & \axiom{CM1T} \\
\ep \leftm x = \dead                                   & \axiom{CM2T} \\
a \seqc x \leftm y = a \seqc (x \parc y)               & \axiom{CM3}  \\
(x \altc y) \leftm z = x \leftm z \altc y \leftm z     & \axiom{CM4}  \\
\ep \commm x = \dead                                   & \axiom{CM5T} \\
x \commm \ep = \dead                                   & \axiom{CM6T} \\
a \seqc x \commm b \seqc y = \commf(a,b) \seqc (x \parc y) 
                                                       & \axiom{CM7}  \\
(x \altc y) \commm z = x \commm z \altc y \commm z     & \axiom{CM8}  \\
x \commm (y \altc z) = x \commm y \altc x \commm z     & \axiom{CM9}  \\
{}                                                                    \\
\encap{H}(\ep) = \ep                                   & \axiom{D0} \\
\encap{H}(a) = a               \hfill \mif a \not\in H & \axiom{D1} \\
\encap{H}(a) = \dead               \hfill \mif a \in H & \axiom{D2} \\
\encap{H}(x \altc y) = \encap{H}(x) \altc \encap{H}(y) & \axiom{D3} \\
\encap{H}(x \seqc y) = \encap{H}(x) \seqc \encap{H}(y) & \axiom{D4} 
\end{axcol}
\end{eqntbl}
\end{table}
In these equations, $a$ and $b$ stand for arbitrary constants of \ACPei\ 
that differ from $\ep$ and $H$ stands for an arbitrary subset of $\Act$.
So, CM3, CM7, and D0--D4 are actually axiom sche\-mas.
Axioms A1--A9, CM1T, CM2T, CM3, CM4, CM5T, CM6T, CM7--CM9, and D0--D4 
are the axioms of \ACPe\ (cf.~\cite[Section 4.4]{BW90}).
Axioms BKS1 and RSP* have been taken from~\cite{BFP01a}.

The iteration operator originates from~\cite{BBP94a}, where it is called 
the binary Kleene star operator.
The unary counterpart of this operator can be defined by the equation 
$x \iter = x \iter \ep$.
From this defining equation, it follows, using RSP*, that 
$x \iter = x \seqc {x \iter} \altc \ep$ and also that
$x \iter y = {x \iter} \seqc y$.

Among the equations derivable from the axioms of \ACPei\ are the 
equations concerning the iteration operator given in 
Table~\ref{deriv-eqns-ACPei}.
\begin{table}[!t]
\caption{Derivable equations for iteration}
\label{deriv-eqns-ACPei}
\begin{eqntbl}
\begin{axcol}
x \iter (y \seqc z) = (x \iter y) \seqc z              & \axiom{BKS2} \\
x \iter (y \seqc ((x \altc y) \iter z) \altc z) = (x \altc y) \iter z
                                                       & \axiom{BKS3} \\
\encap{H}(x \iter y) = \encap{H}(x) \iter \encap{H}(y) & \axiom{BKS4} \\
\ep \iter x = x                                        & \axiom{BKS5}
\end{axcol}
\end{eqntbl}
\end{table}
In the axiom system of \ACPi\ given in~\cite{BBP94a}, the axioms for the 
iteration operator are BKS1--BKS4 instead of BKS1 and RSP*. 
There exist equations derivable from the axioms of \ACPei\ that are not
derivable from the axioms of \ACPei\ with BKS1 and RSP* replaced by 
BKS1--BKS5.
For example, the equation $a \iter \dead = (a \seqc a) \iter \dead$ is 
derivable with BKS1 and RSP*, but not with BKS1--BKS5 
(cf.~\cite{Sew94a}).
Moreover, we do not see how Theorem~\ref{theorem-soundness-HL} of this 
paper can be proved if RSP* is replaced by BKS2--BKS5 (see the remark
following the proof of the theorem).

\section{Data Enriched \ACPei}
\label{sect-deACPei}

In this section, we present \deACPei, data enriched \ACPei.
This extension of \ACPei\ has been inspired by~\cite{BM09d}.
It extends \ACPei\ with features that are relevant to processes in which
data are involved, such as guarded commands (to deal with processes that 
only take place if some data-dependent condition holds), data 
parameterized actions (to deal with process interactions with data 
transfer), and assignment actions (to deal with data that change in the 
course of a process).

In \deACPei, it is assumed that the following has been given with 
respect to data:
\begin{itemize}
\item
a (single- or many-sorted) signature $\sign_\gD$ that includes a 
sort $\Data$ of \emph{data} and constants and/or operators with result 
sort $\Data$;
\item
a minimal algebra $\gD$ of the signature $\sign_\gD$.
\end{itemize}
Moreover, it is assumed that a countably infinite set $\ProgVar$ of 
\emph{flexible variables} has been given.
A flexible variable is a variable whose value may change in the course 
of a process.%
\footnote
{The term flexible variable is used for this kind of variables in 
 e.g.~\cite{Sch97a,Lam94a}.} 
Flexible variables are found under the name program variables in 
imperative programming.
We write $\DataVal$ for the set of all closed terms over the signature
$\sign_\gD$ that are of sort $\Data$.
An \emph{evaluation map} is a function $\sigma$ from $\ProgVar$ to 
$\DataVal \union \ProgVar$ where, for all $v \in \ProgVar$,
$\sigma(v) = v$ if $\sigma(v) \in \ProgVar$. 
Let $\sigma$ be an evaluation map and let $V$ be a finite subset 
of~$\ProgVar$.
Then $\sigma$ is a $V$-\emph{evaluation map} if, for all 
$v \in \ProgVar$, $\sigma(v) \in \DataVal$ iff $v \in V$.

Evaluation maps are intended to provide the data values assigned to 
flexible variables of sort $\Data$ when a term of sort $\Data$ is 
evaluated.
However, in order to fit better in an algebraic setting, they provide 
closed terms over the signature $\sign_\gD$ that denote those data 
values instead.
The requirement that $\gD$ is a minimal algebra guarantees that each 
data value can be represented by a closed term. 
The possibility to map flexible variables to themselves may be used for 
partial evaluation, i.e.\ evaluation where some flexible variables are 
not evaluated. 

The algebraic theory \deACPei\ has three sorts: 
the sort $\Proc$ of \emph{processes},
the sort $\Cond$ of \emph{conditions}, and
the sort $\Data$ of \emph{data}.
\deACPei\ has the constants and operators from $\sign_\gD$ and in 
addition the following constants to build terms of sort~$\Data$:
\begin{itemize}
\item
for each $v \in \ProgVar$, the \emph{flexible variable} constant 
$\const{v}{\Data}$.
\end{itemize}
\deACPei\ has the following constants and operators to build terms of 
sort $\Cond$:
\begin{itemize}
\item
the binary \emph{equality} operator
$\funct{\Leq}{\Data \x \Data}{\Cond}$;
\item
the \emph{truth} constant $\const{\True}{\Cond}$;
\item
the \emph{falsity} constant $\const{\False}{\Cond}$;
\item
the unary \emph{negation} operator $\funct{\Lnot}{\Cond}{\Cond}$;
\item
the binary \emph{conjunction} operator 
$\funct{\Land}{\Cond \x \Cond}{\Cond}$;
\item
the binary \emph{disjunction} operator 
$\funct{\Lor}{\Cond \x \Cond}{\Cond}$;
\item
the binary \emph{implication} operator 
$\funct{\Limpl}{\Cond \x \Cond}{\Cond}$;
\item
the unary variable-binding \emph{universal quantification} operator 
$\funct{\forall}{\Cond}{\Cond}$ that binds a variable \linebreak[2] of 
sort $\Data$; 
\item
the unary variable-binding \emph{existential quantification} operator 
$\funct{\exists}{\Cond}{\Cond}$ that binds a variable of sort $\Data$. 
\end{itemize}
\deACPei\ has the constants and operators of \ACPei\ and in addition the 
following operators to build terms of sort $\Proc$:
\begin{itemize}
\item
the binary \emph{guarded command} operator 
$\funct{\gc}{\Cond \x \Proc}{\Proc}$;
\item
for each $n \in \Nat$, for each $a \in \Act$,
the $n$-ary \emph{data parameterized action} operator
$\funct{a}
 {\underbrace{\Data \x \cdots \x \Data}_{n\; \mathrm{times}}}{\Proc}$;
\item
for each $v \in \ProgVar$, 
a unary \emph{assignment action} operator
$\funct{\assop{v}}{\Data}{\Proc}$;
\item
for each evaluation map $\sigma$, 
a unary \emph{evaluation} operator 
$\funct{\eval{\sigma}}{\Proc}{\Proc}$.
\end{itemize}

We assume that there are countably infinite sets of variables of sort 
$\Cond$ and $\Data$ and that the sets of variables of sort $\Proc$, 
$\Cond$, and $\Data$ are mutually disjoint and disjoint from $\ProgVar$.
The formation rules for terms are the usual ones for the many-sorted 
case (see e.g.~\cite{ST99a,Wir90a}) and in addition the following rule:
\begin{itemize}
\item
if $O$ is a variable-binding operator 
$\funct{O}{S_1 \x \ldots \x S_n}{S}$ that binds a variable of sort~$S'$,
$t_1,\ldots,t_n$~are terms of sorts $S_1,\ldots,S_n$, respectively, and 
$X$ is a variable of sort $S'$, then $O X (t_1,\ldots,t_n)$ is a term of 
sort $S$ (cf.~\cite{PS95a}).
\end{itemize}
We use the same notational conventions as before.
We also use infix notation for the additional binary operators.
Moreover, we use the notation $\ass{v}{e}$, where $v \in \ProgVar$ and 
$e$ is a term of sort $\Data$, for the term $\assop{v}(e)$.

We use the notation  $\phi \Liff \psi$, where $\phi$ and $\psi$ are 
terms of sort $\Cond$, for the term
$(\phi \Limpl \psi) \Land (\psi \Limpl \phi)$.
Moreover, we use the notation $\LOR \Phi$, 
where $\Phi = \set{\phi_1,\ldots,\phi_n}$ and
$\phi_1, \ldots, \phi_n$ are terms of sort $\Cond$,
for the term $\phi_1 \Lor \ldots \Lor \phi_n$. 

We write 
$\ProcTerm$ for the set of all closed terms of sort $\Proc$,
$\CondTerm$ for the set of all closed terms of sort $\Cond$, and
$\DataTerm$ for the set of all closed terms of sort~$\Data$.

Each term from $\CondTerm$ can be taken as a formula of a first-order 
language with equality of $\gD$ by taking the flexible variable 
constants as additional variables of sort $\Data$.
We implicitly take the flexible variable constants as additional 
variables of sort $\Data$ wherever the context asks for a formula.
In this way, each term from $\CondTerm$ can be interpreted as a
formula in $\gD$.
The axioms of \deACPei\ (given below) include an equation $\phi = \psi$ 
for each two terms $\phi$ and $\psi$ from $\CondTerm$ for which the 
formula $\phi \Liff \psi$ holds in $\gD$.

Let $p$ be a term from $\ProcTerm$, $\phi$ be a term from $\CondTerm$, 
and $e_1,\ldots,e_n$ and $e$ be terms from $\DataTerm$. 
Then the additional operators can be explained as follows:
\begin{itemize}
\item
the term $\phi \gc p$ denotes the process that behaves as the process 
denoted by $p$ under condition $\phi$;
\item
\sloppy
the term $a(e_1,\ldots,e_n)$ denotes the process that is only capable of 
first performing action $a(e_1,\ldots,e_n)$ and next terminating 
successfully;
\item
the term $\ass{v}{e}$ denotes the process that is only capable of first 
performing action $\ass{v}{e}$, whose intended effect is the assignment 
of the result of evaluating $e$ to flexible variable $v$, and next 
terminating successfully; 
\item
the term $\eval{\sigma}(p)$ denotes the process that behaves the same as 
the process denoted by $p$ except that each subterm of $p$ that belongs 
to $\DataTerm$ is evaluated using the evaluation map $\sigma$ updated 
according to the assignment actions that have taken place at the point 
where the subterm is encountered.
\end{itemize}
Evaluation operators are a variant of state operators 
(see e.g.~\cite{BB88}).

The guarded command operator is often used to construct \deACPei\ terms 
that are reminiscent of control flow statements of imperative 
programming languages.
For example, terms of the form 
$\phi \gc t \altc (\Lnot\, \phi) \gc t'$ 
are reminiscent of if-then-else statements and terms of the form 
$(\phi \gc t) \iter ((\Lnot\, \phi) \gc \ep)$ 
are reminiscent of while-do statements.
The following \deACPei\ term contains a subterm of the latter form 
($i, j, q, r \in \ProgVar$):
\begin{ldispl}
\ass{q}{0} \seqc \ass{r}{i} \seqc 
(((r \geq j) \gc \ass{q}{q + 1} \seqc \ass{r}{r - j}) \iter 
 ((\Lnot\; r \geq j) \gc \ep))\;.
\end{ldispl}%
This term is reminiscent of a program that computes the quotient and 
remainder of dividing two integers by repeated subtraction.
That is, the final values of $q$ and $r$ are the quotient and remainder 
of dividing the initial value of $i$ by the initial value of $j$.
An evaluation operator can be used to show that this is the case for
given initial values of $i$ and $j$.
For example, consider the case where the initial values of $i$ and $j$
are $11$ and $3$, respectively.
Let $\sigma$ be an evaluation map such that $\sigma(i) = 11$ and
$\sigma(j) = 3$.
Then the following equation can be derived from the axioms of \deACPei\
given below:
\begin{ldispl}
\eval{\sigma}
(\ass{q}{0} \seqc \ass{r}{i} \seqc 
(((r \geq j) \gc \ass{q}{q + 1} \seqc \ass{r}{r - j}) \iter 
 ((\Lnot\; r \geq j) \gc \ep))) \\
\; {} =
\ass{q}{0} \seqc \ass{r}{11} \seqc \ass{q}{1} \seqc \ass{r}{8} \seqc 
\ass{q}{2} \seqc \ass{r}{5} \seqc \ass{q}{3} \seqc \ass{r}{2}\;. 
\end{ldispl}%
This equation shows that in the case where the initial values of $i$ 
and $j$ are $11$ and $3$ the final values of $q$ and $r$ are $3$ and 
$2$ (which are the quotient and remainder of dividing $11$ by $3$).

An evaluation map $\sigma$ can be extended homomorphically from flexible
variables to terms of sort $\Data$ and terms of sort $\Cond$.
These extensions are denoted by $\sigma$ as well.
We write $\sigma\mapupd{e}{v}$ for the evaluation map $\sigma'$ defined 
by $\sigma'(v') = \sigma(v')$ if $v' \not\equiv v$ and $\sigma'(v) = e$.

The axioms of \deACPei\ are the axioms of \ACPei\ and in addition the 
equations given in Table~\ref{axioms-deACPei}.%
\begin{table}[!p]
\caption{Axioms of \deACPei}
\label{axioms-deACPei}
\begin{eqntbl}
\begin{saxcol}
e = e'         & \mif \Sat{\gD}{\fol{e = e'}}          & \axiom{IMP1} \\
\phi = \psi    & \mif \Sat{\gD}{\fol{\phi \Liff \psi}} & \axiom{IMP2} \\
{}                                                                    \\
\True \gc x = x                                      & & \axiom{GC1}  \\
\False \gc x = \dead                                 & & \axiom{GC2}  \\
\phi \gc \dead = \dead                               & & \axiom{GC3}  \\
\phi \gc (x \altc y) = \phi \gc x \altc \phi \gc y   & & \axiom{GC4}  \\
\phi \gc x \seqc y = (\phi \gc x) \seqc y            & & \axiom{GC5}  \\
\phi \gc (\psi \gc x) = (\phi \Land \psi) \gc x      & & \axiom{GC6}  \\
(\phi \Lor \psi) \gc x = \phi \gc x \altc \psi \gc x & & \axiom{GC7}  \\
(\phi \gc x) \leftm y = \phi \gc (x \leftm y)        & & \axiom{GC8}  \\
(\phi \gc x) \commm y = \phi \gc (x \commm y)        & & \axiom{GC9}  \\
x \commm (\phi \gc y) = \phi \gc (x \commm y)        & & \axiom{GC10} \\
\encap{H}(\phi \gc x) = \phi \gc \encap{H}(x)        & & \axiom{GC11} \\
{}                                                                    \\
\eval{\sigma}(\ep) = \ep                             & & \axiom{V0}   \\
\eval{\sigma}(a \seqc x) = a \seqc \eval{\sigma}(x)  & & \axiom{V1}   \\
\eval{\sigma}(a(e_1,\ldots,e_n) \seqc x) = 
a(\sigma(e_1),\ldots,\sigma(e_n)) \seqc \eval{\sigma}(x)
                                                     & & \axiom{V2}   \\
\eval{\sigma}(\ass{v}{e} \seqc x) = 
{\assd{v}{\sigma(e)} \seqc \eval{\sigma\mapupd{\sigma(e)}{v}}(x)} 
                                                     & & \axiom{V3}   \\
\eval{\sigma}(x \altc y) = \eval{\sigma}(x) \altc \eval{\sigma}(y)
                                                     & & \axiom{V4}   \\
\eval{\sigma}(\phi \gc y) = \sigma(\phi) \gc \eval{\sigma}(x)
                                                     & & \axiom{V5}   \\
{}                                                                    \\
a(e_1,\ldots,e_n) \seqc x \leftm y = a(e_1,\ldots,e_n) \seqc (x \parc y)
                                                    & & \axiom{CM3D}  \\
a(e_1,\ldots,e_n) \seqc x \commm b(e'_1,\ldots,e'_n) \seqc y = 
 {} \\ \quad
(e_1 = e'_1 \Land \ldots \Land e_n = e'_n) \gc c(e_1,\ldots,e_n) \seqc
(x \parc y)                    & \mif \commf(a,b) = c & \axiom{CM7Da} \\
a(e_1,\ldots,e_n) \seqc x \commm b(e'_1,\ldots,e'_m) \seqc y = \dead
  & \mif \commf(a,b) = \dead \;\mathrm{or}\; n \neq m & \axiom{CM7Db} \\
a(e_1,\ldots,e_n) \seqc x \commm b \seqc y = \dead  & & \axiom{CM7Dc} \\
a \seqc x \commm b(e_1,\ldots,e_n) \seqc y = \dead  & & \axiom{CM7Dd} \\
\encap{H}(a(e_1,\ldots,e_n)) = a(e_1,\ldots,e_n) 
                                     & \mif a \not\in H & \axiom{D1D} \\
\encap{H}(a(e_1,\ldots,e_n)) = \dead & \mif a \in H     & \axiom{D2D} \\
{}                                                                    \\
\ass{v}{e} \seqc x \leftm y = \ass{v}{e} \seqc (x\parc y)
                                                    & & \axiom{CM3A}  \\
\ass{v}{e} \seqc x \commm y = \dead                 & & \axiom{CM5A}  \\
x \commm \ass{v}{e} \seqc y = \dead                 & & \axiom{CM6A}  \\
\encap{H}(\ass{v}{e}) = \ass{v}{e}                  & & \axiom{D1A}    
\end{saxcol}
\end{eqntbl}
\end{table}
In these equations, 
$\phi$ and $\psi$ stand for arbitrary terms from $\CondTerm$, 
$e$, $e_1,e_2,\ldots$, and $e'$, $e'_1,e'_2,\ldots$ stand for arbitrary 
terms from $\DataTerm$, 
$v$ stands for an arbitrary flexible variable from $\ProgVar$,
$\sigma$ stands for an arbitrary evaluation map,
$a$ and $b$ stand for arbitrary constants of \deACPei\ that differ from 
$\ep$, 
$c$ stands for an arbitrary constant of \deACPei\ that differ from 
$\ep$ and $\dead$, and
$H$ stands for an arbitrary subset of $\Act$.
Axioms GC1--GC11 have been taken from~\cite{BB92c} (using a different 
numbering), but with the axioms with occurrences of conditional 
expressions of the form $\cond{p}{\phi}{q}$ replaced by simpler axioms.
Axioms CM3D, CM7Da, CM7Db, D1D, and D2D have been inspired 
by~\cite{BM09d}.

The set $\AProcTerm$ of \emph{actions} of \deACPei\ is inductively 
defined by the following rules:
\begin{itemize}
\item
if $a \in \Act$, then $a \in \AProcTerm$;
\item
if $a \in \Act$ and $e_1,\dots,e_n \in \DataTerm$, then
$a(e_1,\dots,e_n) \in \AProcTerm$;
\item
if $v \in \ProgVar$ and $e \in \DataTerm$, then
$\ass{v}{e} \in \AProcTerm$.
\end{itemize}
The elements of $\AProcTerm$ are the processes that are considered to be 
atomic.

The set $\HNF$ of \emph{head normal forms} of \deACPei\ is inductively 
defined by the following rules:
\begin{itemize}
\item 
$\dead \in \HNF$;
\item 
if $\phi \in \CondTerm$, then $\phi \gc \ep \in \HNF$;
\item 
if $\phi \in \CondTerm$, $\alpha \in \AProcTerm$, and $p \in \ProcTerm$, 
then $\phi \gc a \seqc p \in \HNF$;
\item 
if $p,p' \in \HNF$, then $p \altc p' \in \HNF$.
\end{itemize}
The following lemma about head normal forms is used in later sections.
\begin{lemma}
\label{lemma-HNF}
For all terms $p \in \ProcTerm$, there exists a term $q \in \HNF$ such 
that $p = q$ is derivable from the axioms of \deACPei.
\end{lemma}
\begin{proof}
This is straightforwardly proved by induction on the structure of $p$.
The cases where $p$ is of the form $\dead$, $\ep$ or $\alpha$
($\alpha \in \AProcTerm$) are trivial.
The case where $p$ is of the form $p_1 \altc p_2$ follows immediately
from the induction hypothesis. 
The case where $p$ is of the form $p_1 \parc p_2$ follows immediately
from the case that $p$ is of the form $p_1 \leftm p_2$ and the case that
$p$ is of the form $p_1 \commm p_2$.
Each of the other cases follow immediately from the induction hypothesis
and a claim that is easily proved by structural induction.
In the case where $p$ is of the form $p_1 \commm p_2$, each of the cases
to be considered in the inductive proof demands an additional proof by
structural induction.
\qed
\end{proof}

Some earlier extensions of \ACP\ include Hoare's ternary counterpart of 
the binary guarded command operator (see e.g.~\cite{BB92c}).  
This operator can be defined by the equation
$\cond{x}{u}{y} = u \gc x \altc (\Lnot\, u) \gc y$. 
From this defining equation, it follows that  
$u \gc x = \cond{x}{u}{\dead}$.
In~\cite{GP94b}, a unary counterpart of the binary guarded command 
operator is used.
This operator can be defined by the equation $\guard{u} = u \gc \ep$.
From this defining equation, it follows that  
$u \gc x = \guard{u} \seqc x$ and also that $\guard{\True} = \ep$ and
$\guard{\False} = \dead$.
In~\cite{GP94b}, the processes denoted by closed terms of the form 
$\guard{\phi}$ are called guards.


\section{Structural Operational Semantics and Bisimulation Equivalence}
\label{sect-sos-bisim}

In this section, we present a structural operational semantics of 
\deACPei, define a notion of bisimulation equivalence based on this 
semantics, and show that the axioms of \deACPei\ are sound with respect 
to this bisimulation equivalence.

We write $\sCondTerm$ for the set of all terms $\phi \in \CondTerm$ for 
which $\nSat{\gD}{\phi \Liff \False}$.
As formulas of a first-order language with equality of $\gD$, the terms 
from $\sCondTerm$ are the formulas that are satisfiable in $\gD$.

We start with the presentation of the structural operational semantics 
of \deACPei.
The following transition relations on $\ProcTerm$ are used:
\begin{itemize}
\item 
for each $\phi \in \sCondTerm$,
a unary relation ${\sterm{\phi}}$;
\item 
for each $\ell \in \sCondTerm \x \AProcTerm$,
a binary relation ${\step{\ell}}$.
\end{itemize}

We write $\isterm{p}{\phi}$ instead of 
$p \in {\sterm{\phi}}$ and
$\astep{p}{\gact{\phi}{\alpha}}{q}$ instead of 
$\tup{p,q} \in {\step{\tup{\phi,\alpha}}}$. 
The relations ${\sterm{\phi}}$ and ${\step{\ell}}$ can be explained as 
follows:
\begin{itemize}
\item
$\isterm{p}{\phi}$: 
$p$ is capable of terminating successfully under condition $\phi$;
\item
$\astep{p}{\gact{\phi}{\alpha}}{q}$: 
$p$ is capable of performing action $\alpha$ under condition $\phi$ and 
then proceeding as $q$.
\end{itemize}

The structural operational semantics of \deACPei\ is described by the 
transition rules given in Table~\ref{sos-deACPei}.%
\begin{table}[!p]
\caption{Transition rules for \deACPei}
\label{sos-deACPei}
\begin{ruletbl}
{} \\[-3ex]
\Rule
{}
{\isterm{\ep}{\True}}
\\[-2ex]
\Rule
{}
{\astep{a}{\gact{\True}{a}}{\ep}}
\quad\;
\Rule
{}
{\astep{a(e_1,\ldots,e_n)}{\gact{\True}{a(e_1,\ldots,e_n)}}{\ep}}
\quad\;
\Rule
{}
{\astep{\ass{v}{e}}{\gact{\True}{\ass{v}{e}}}{\ep}}
\\
\Rule
{\isterm{x}{\phi}}
{\isterm{x \altc y}{\phi}}
\quad\;
\Rule
{\isterm{y}{\phi}}
{\isterm{x \altc y}{\phi}}
\quad\;
\Rule
{\astep{x}{\gact{\phi}{\alpha}}{x'}}
{\astep{x \altc y}{\gact{\phi}{\alpha}}{x'}}
\quad\;
\Rule
{\astep{y}{\gact{\phi}{\alpha}}{y'}}
{\astep{x \altc y}{\gact{\phi}{\alpha}}{y'}}
\\
\RuleC
{\isterm{x}{\phi},\; \isterm{y}{\psi}}
{\isterm{x \seqc y}{\phi \Land \psi}}
{\nSat{\gD}{\fol{\phi \Land \psi \Liff \False}}}
\quad\;
\RuleC
{\isterm{x}{\phi},\; \astep{y}{\gact{\psi}{\alpha}}{y'}}
{\astep{x \seqc y}{\gact{\phi \Land \psi}{\alpha}}{y'}}
{\nSat{\gD}{\fol{\phi \Land \psi \Liff \False}}}
\quad\;
\Rule
{\astep{x}{\gact{\phi}{\alpha}}{x'}}
{\astep{x \seqc y}{\gact{\phi}{\alpha}}{x' \seqc y}}
\\
\Rule
{\isterm{y}{\phi}}
{\isterm{x \iter y}{\phi}}
\quad\;
\Rule
{\astep{y}{\gact{\phi}{\alpha}}{y'}}
{\astep{x \iter y}{\gact{\phi}{\alpha}}{y'}}
\quad\;
\Rule
{\astep{x}{\gact{\phi}{\alpha}}{x'}}
{\astep{x \iter y}{\gact{\phi}{\alpha}}{x' \seqc (x \iter y)}}
\\
\RuleC
{\isterm{x}{\phi}}
{\isterm{\psi \gc x}{\phi \Land \psi}}
{\nSat{\gD}{\fol{\phi \Land \psi \Liff \False}}}
\quad\;
\RuleC
{\astep{x}{\gact{\phi}{\alpha}}{x'}}
{\astep{\psi \gc x}{\gact{\phi \Land \psi}{\alpha}}{x'}}
{\nSat{\gD}{\fol{\phi \Land \psi \Liff \False}}}
\\
\RuleC
{\isterm{x}{\phi},\; \isterm{y}{\psi}}
{\isterm{x \parc y}{\phi \Land \psi}}
{\nSat{\gD}{\fol{\phi \Land \psi \Liff \False}}}
\quad\;
\Rule
{\astep{x}{\gact{\phi}{\alpha}}{x'}}
{\astep{x \parc y}{\gact{\phi}{\alpha}}{x' \parc y}}
\quad\;
\Rule
{\astep{y}{\gact{\phi}{\alpha}}{y'}}
{\astep{x \parc y}{\gact{\phi}{\alpha}}{x \parc y'}}
\\
\RuleC
{\astep{x}{\gact{\phi}{a}}{x'},\; \astep{y}{\gact{\psi}{b}}{y'}}
{\astep{x \parc y}{\gact{\phi \Land \psi}{c}}{x' \parc y'}}
{\commf(a,b) = c,\; \nSat{\gD}{\fol{\phi \Land \psi \Liff \False}}}
\\
\RuleC
{\astep{x}{\gact{\phi}{a(e_1,\ldots,e_n)}}{x'},\; 
 \astep{y}{\gact{\psi}{b(e'_1,\ldots,e'_n)}}{y'}}
{\astep{x \parc y}
  {\gact{\phi \Land \psi \Land e_1 = e'_1 \Land \ldots \Land  e_n = e'_n}
  {c(e_1,\ldots,e_n)}}{x' \parc y'}}
{\begin{array}{@{}c@{}}
 \commf(a,b) = c, \\[.25ex] 
 \nSat{\gD}
  {\fol{\phi \Land \psi \Land e_1 = e'_1 \Land \ldots \Land  e_n = e'_n \Liff
        \False}}
 \end{array}
}
\\
\Rule
{\astep{x}{\gact{\phi}{\alpha}}{x'}}
{\astep{x \leftm y}{\gact{\phi}{\alpha}}{x' \parc y}}
\\
\RuleC
{\astep{x}{\gact{\phi}{a}}{x'},\; \astep{y}{\gact{\psi}{b}}{y'}}
{\astep{x \commm y}{\gact{\phi \Land \psi}{c}}{x' \parc y'}}
{\commf(a,b) = c,\; \nSat{\gD}{\fol{\phi \Land \psi \Liff \False}}}
\\
\RuleC
{\astep{x}{\gact{\phi}{a(e_1,\ldots,e_n)}}{x'},\; 
 \astep{y}{\gact{\psi}{b(e'_1,\ldots,e'_n)}}{y'}}
{\astep{x \commm y}
  {\gact{\phi \Land \psi \Land e_1 = e'_1 \Land \ldots \Land  e_n = e'_n}
  {c(e_1,\ldots,e_n)}}{x' \parc y'}}
{\begin{array}{@{}c@{}}
 \commf(a,b) = c, \\[.25ex]
 \nSat{\gD}
  {\fol{\phi \Land \psi \Land e_1 = e'_1 \Land \ldots \Land  e_n = e'_n \Liff
        \False}}
 \end{array}
}
\\
\Rule
{\isterm{x}{\phi}}
{\isterm{\encap{H}(x)}{\phi}}
\quad\;
\RuleC
{\astep{x}{\gact{\phi}{a}}{x'}}
{\astep{\encap{H}(x)}{\gact{\phi}{a}}{\encap{H}(x')}}
{a \not\in H}
\quad\;
\RuleC
{\astep{x}{\gact{\phi}{a(e_1,\ldots,e_n)}}{x'}}
{\astep{\encap{H}(x)}{\gact{\phi}{a(e_1,\ldots,e_n)}}{\encap{H}(x')}}
{a \not\in H}
\\
\Rule
{\astep{x}{\gact{\phi}{\ass{v}{e}}}{x'}}
{\astep{\encap{H}(x)}{\gact{\phi}{\ass{v}{e}}}{\encap{H}(x')}}
\\
\Rule
{\isterm{x}{\phi}}
{\isterm{\eval{\sigma}(x)}{\sigma(\phi)}}
\quad\;
\Rule
{\astep{x}{\gact{\phi}{a}}{x'}}
{\astep{\eval{\sigma}(x)}{\gact{\sigma(\phi)}{a}}{\eval{\sigma}(x')}}
\quad\;
\Rule
{\astep{x}{\gact{\phi}{a(e_1,\ldots,e_n)}}{x'}}
{\astep{\eval{\sigma}(x)}
       {\gact{\sigma(\phi)}{a(\sigma(e_1),\ldots,\sigma(e_n))}}
       {\eval{\sigma}(x')}}
\\
\RuleC
{\astep{x}{\gact{\phi}{\ass{v}{e}}}{x'}}
{\astep{\eval{\sigma}(x)}{\gact{\sigma(\phi)}{\ass{v}{\sigma(e)}}}
       {\eval{\sigma\mapupd{\sigma(e)}{v}}(x')}}
{\sigma(v) \in \DataVal}
\quad\;
\RuleC
{\astep{x}{\gact{\phi}{\ass{v}{e}}}{x'}}
{\astep{\eval{\sigma}(x)}{\gact{\sigma(\phi)}{\ass{v}{\sigma(e)}}}
       {\eval{\sigma}(x')}}
{\sigma(v) \notin \DataVal}
\end{ruletbl}
\end{table}
In this table, 
$a$, $b$, and $c$ stand for arbitrary basic actions from $\Act$,
$v$ stands for an arbitrary flexible variable from $\ProgVar$,
$e$ and $e_1,e_2,\ldots$ stand for arbitrary terms from $\DataTerm$,
$\phi$ and $\psi$ stand for arbitrary terms from $\sCondTerm$,
$\alpha$ stands for an arbitrary term from $\AProcTerm$,
$H$ stands for arbitrary subset of $\Act$, and
$\sigma$ stands for an arbitrary evaluation map.

Two process are considered equal if they can simulate each other.
In order to make this precise, we will define the notion of 
bisimulation equivalence on the set $\ProcTerm$ below.
In the definition concerned, we need an equivalence relation on the set 
$\AProcTerm$.

Two actions $\alpha,\alpha' \in \AProcTerm$ are \emph{data equivalent},
written $\alpha \simeq \alpha'$, iff one of the following holds:
\begin{itemize}
\item
there exists an $a \in \Act$ such that $\alpha = a$ and $\alpha' = a$;
\item
there exist an $a \in \Act$ and 
$e_1,\dots,e_n,e'_1,\dots,e'_n \in \DataTerm$
such that 
$\Sat{\gD}
  {\fol{e_1 = e'_1 \Land\linebreak[2] \ldots \Land  e_n = e'_n}}$,
$\alpha = a(e_1,\dots,e_n)$, and $\alpha' = a(e'_1,\dots,e'_n)$;
\item
there exist a $v \in \ProgVar$ and $e,e' \in \DataTerm$ such that 
$\Sat{\gD}{\fol{e = e'}}$, $\alpha = \ass{v}{e}$, and 
$\alpha' = \ass{v}{e'}$.
\end{itemize}
We write $[\alpha]$, where $\alpha \in \AProcTerm$, for the equivalence 
class of $\alpha$ with respect to $\simeq$.

A \emph{bisimulation} is a binary relation $R$ on $\ProcTerm$ such that, 
for all terms $p,q \in \ProcTerm$ with $(p,q) \in R$, the following 
conditions hold:
\begin{itemize}
\item
if $\astep{p}{\gact{\phi}{\alpha}}{p'}$, then 
there exists a finite set $\Psi \subseteq \sCondTerm$ such that 
$\Sat{\gD}{\phi \Limpl \LOR \Psi}$ and, 
for all $\psi \in \Psi$, there exist an $\alpha' \in [\alpha]$ and a 
$q' \in \ProcTerm$ such that 
\smash{$\astep{q}{\gact{\psi}{\alpha'}}{q'}$} and $(p',q') \in R$;
\item
if $\astep{q}{\gact{\phi}{\alpha}}{q'}$, then 
there exists a finite set $\Psi \subseteq \sCondTerm$ such that 
$\Sat{\gD}{\phi \Limpl \LOR \Psi}$ and, 
for all $\psi \in \Psi$, there exist an $\alpha' \in [\alpha]$ and a 
$p' \in \ProcTerm$ such that 
\smash{$\astep{p}{\gact{\psi}{\alpha'}}{p'}$} and $(p',q') \in R$;
\item
if $\isterm{p}{\phi}$, then 
there exists a finite set $\Psi \subseteq \sCondTerm$ such that 
$\Sat{\gD}{\phi \Limpl \LOR \Psi}$ and, 
for all $\psi \in \Psi$, $\isterm{q}{\psi}$;
\item
if $\isterm{q}{\phi}$, then 
there exists a finite set $\Psi \subseteq \sCondTerm$ such that 
$\Sat{\gD}{\phi \Limpl \LOR \Psi}$ and, 
for all $\psi \in \Psi$, $\isterm{p}{\psi}$.
\end{itemize}
Two terms $p,q \in \ProcTerm$ are \emph{bisimulation equivalent}, 
written $p \bisim q$, if there exists a bisimulation $R$ such that 
$(p,q) \in R$.
Let $R$ be a bisimulation such that $(p,q) \in R$.
Then we say that $R$ is a bisimulation \emph{witnessing} $p \bisim q$.

The above definition of a bisimulation deviates from the standard 
definition because a transition on one side may be simulated by a set
of transitions on the other side. 
For example, 
the transition
\smash{$\astep{(\phi_1 \Lor \phi_2) \gc a \seqc b}
        {\gact{\phi_1 \Lor \phi_2}{a}}{b}$}
is simulated by
the set of transitions consisting of
\smash{$\astep{\phi_1 \gc a \seqc b}{\gact{\phi_1}{a}}{b}$} and 
\smash{$\astep{\phi_2 \gc a \seqc b}{\gact{\phi_2}{a}}{b}$}.
A bisimulation as defined above is called a \emph{splitting}
bisimulation in~\cite{BM05a}.

Bisimulation equivalence is a congruence with respect to the operators 
of \deACPei\ of which the result sort and at least one argument sort is 
$\Proc$.
\begin{theorem}[Congruence]
\label{theorem-congruence-ACPei}
For all terms $p,q,p',q' \in \ProcTerm$ and all terms 
$\phi \in \CondTerm$, 
$p \bisim p'$ and $q \bisim q'$ only if 
$p \altc q \bisim p' \altc q'$, 
$p \seqc q \bisim p' \seqc q'$,
$p \iter q \bisim p' \iter q'$,
$\phi \gc p \bisim \phi \gc p'$, 
$p \parc q \bisim p' \parc q'$, 
$p \leftm q \bisim p' \leftm q'$,
$p \commm q \bisim p' \commm q'$,
$\encap{H}(p) \bisim \encap{H}(p')$, and
$\eval{\sigma}(p) \bisim \eval{\sigma}(p')$.
\end{theorem}
\begin{proof} 
We can reformulate the transition rules such that:
\begin{itemize}
\item
bisimulation equivalence based on the reformulated transition rules 
according to the standard definition of bisimulation equivalence 
coincides with bisimulation equivalence based on the original transition 
rules according to the definition of bisimulation equivalence given 
above;
\item
the reformulated transition rules make up a transition system 
specification in path format.
\end{itemize}
The reformulation is similar to the one for the transition rules for 
BPAps outlined in~\cite{BB94b}.
The proposition follows now immediately from the well-known result that
bisimulation equivalence according to the standard definition of 
bisimulation equivalence is a congruence if the transition rules 
concerned make up a transition system specification in path format 
(see e.g.~\cite{BV93a}).
\qed
\end{proof}
The underlying idea of the reformulation referred to above is that we 
replace each transition \smash{$\astep{p}{\gact{\phi}{\alpha}}{p'}$} by 
a transition \smash{$\astep{p}{\gact{\nu}{[\alpha]}}{p'}$} for each 
valuation of variables $\nu$ such that $\Sat{\gD}{\phi}\,[\nu]$, and 
likewise \smash{$\isterm{p}{\phi}$}.
Thus, in a bisimulation, a transition on one side must be simulated by a
single transition on the other side.
We did not present the reformulated structural operational semantics
in this paper because it is, in our opinion, intuitively less appealing.  

The axioms of \deACPei\ are sound with respect to $\bisim$ for equations 
between terms from $\ProcTerm$.
\begin{theorem}[Soundness]
\label{theorem-soundness-ACPei}
For all terms $p,q \in \ProcTerm$, $p = q$ is derivable from the axioms
of \deACPei\ only if $p \bisim q$.
\end{theorem}
\begin{proof} 
Because ${\bisim}$ is a congruence, it is sufficient to prove the 
theorem for all substitution instances of each axiom of \deACPei.
We will loosely say that a relation contains all closed substitution 
instances of an equation if it contains all pairs $(p,q)$ such that 
$p = q$ is a closed substitution instance of the equation.

For each axiom, we can construct a bisimulation $R$ witnessing 
$p \bisim q$ for all closed substitution instances $p = q$ of the axiom 
as follows:
\begin{itemize}
\item
in the case of A1--A6, A8, A9, BKS1, CM3, CM4, CM7--CM9, D1, D3, D4, 
GC1, GC4--GC11, V1--V5, CM3D, CM7Da, D1D, CM3A, and D1A,
we take the relation $R$ that consists of all closed substitution 
instances of the axiom concerned and the equation $x = x$;
\item
in the case of A7, CM2T, CM5T, CM6T, D0, D2, GC2, GC3, V0, CM7Db--CM7Dd,
D2D, CM5A, and CM6A,
we take the relation $R$ that consists of all closed substitution 
instances of the axiom concerned;
\item
in the case of CM1T, we take the relation $R$ that consists of all 
closed substitution instances of CM1T, the equation 
$x \parc y = y \parc x$, and the equation $x = x$; 
\item
in the case of RSP*, we take the relation $R$ that consists of all 
closed substitution instances $r = p \iter q$ of the consequent of RSP*
for which $r \bisim p \seqc r \altc q$ and all closed substitution 
instances of the equation $x = x$.
\qed
\end{itemize}
\end{proof}

We have not been able to prove the completeness of the axioms of 
\deACPei\ with respect to $\bisim$ for equations between terms from 
$\ProcTerm$.
Such a proof would give an affirmative answer to an open question about 
the axiomatization of the iteration operator already posed in 1984 by 
Milner~\cite[page~465]{Mil84a}.
Until now, all attempts to answer this question have failed 
(see~\cite{Fok96a}).

\section{A Hoare Logic of Asserted Processes}
\label{sect-HL}

In this section, we present \HL, a Hoare logic of asserted processes 
based on \deACPei, define what it means that an asserted process is 
true, and show that the axioms and rules of this logic are sound with 
respect to this meaning.

We write $\HProcTerm$ for the set of all closed terms of sort $\Proc$ in 
which the evaluation operators $\eval{\sigma}$ and the auxiliary 
operators $\leftm$ and $\commm$ do not occur and we write $\HCondTerm$ 
for the set of all terms of sort $\Cond$ in which variables of sort 
$\Cond$ do not occur.
Clearly, $\HProcTerm \subset \ProcTerm$ and 
$\CondTerm \subset \HCondTerm$.

An \emph{asserted process} is a formula of the form 
$\assproc{\phi}{p}{\psi}$, 
where $p \in \HProcTerm$ and $\phi,\psi \in \HCondTerm$.
Here, $\phi$ is called the \emph{pre-condition} of the asserted process 
and $\psi$ is called the \emph{post-condition} of the asserted process.

The intuitive meaning of an asserted process $\assproc{\phi}{p}{\psi}$
is as follows: if $\phi$ holds at the start of $p$ and $p$ eventually 
terminates successfully, then $\psi$ holds at the successful termination 
of $p$.
The conditions $\phi$ and $\psi$ concern the data values assigned to 
flexible variables at the start and at successful termination, 
respectively.
Therefore, in general, one or more flexible variables occur in $\phi$ 
and~$\psi$.
Unlike in $p$, (logical) variables of sort $\Data$ may also occur in 
$\phi$ and $\psi$. 
This allows of referring in $\psi$ to the data values assigned to 
flexible variables at the start, like in 
$\assproc{v = u}{\ass{v}{v+1}}{v = u + 1}$.

Below, we use the notion of equivalence under $V$-evaluation to make the 
intuitive meaning of asserted processes more precise.

We write $\FVar(p)$, where $p \in \ProcTerm$, for the set of all 
$v \in \ProgVar$ that occur in $p$ and
likewise $\FVar(\phi)$, where $\phi \in \HCondTerm$, for the set of all 
$v \in \ProgVar$ that occur in $\phi$. \linebreak[2]
We write $\FAss(p)$, where $p \in \ProcTerm$, for the set of all 
$v \in \FVar(p)$ that occur in subterms of $p$ that are of the form
$\ass{v}{e}$.
Moreover, we write $\ProcTerm_V$, where $V$ is a finite subset of 
$\ProgVar$, for the set 
$\set{p \in \ProcTerm \where \FVar(p) \subseteq V}$.

Let $V$ be a finite subset of $\ProgVar$ and let $p,q \in \ProcTerm_V$.
Then $p$ and~$q$ are \emph{equivalent under $V\!$-evaluation}, written
$p \evaleqv{V} q$, if, for all $V$-evaluation maps $\sigma$, 
$\eval{\sigma}(p) = \eval{\sigma}(q)$ 
is derivable from the axioms of \deACPei.

Notice that $\evaleqv{V}$, where $V$ be a finite subset of $\ProgVar$, 
is an equivalence relation indeed.
Notice further that, for all $p,q \in \ProcTerm_W$, $W \subset V$ and 
$p \evaleqv{W} q$ only if $p \evaleqv{V} q$.

Let $\assproc{\phi}{p}{\psi}$ be an asserted process and 
let $V = \FVar(\phi) \union \FVar(p) \union \FVar(\psi)$.
Then $\assproc{\phi}{p}{\psi}$ is \emph{true} if, 
for all closed substitution instances $\assproc{\phi'}{p}{\psi'}$ of 
$\assproc{\phi}{p}{\psi}$,
$\phi' \gc p \evaleqv{V} (\phi' \gc p) \seqc (\psi' \gc \ep)$.

To justify the claim that the definition given above reflects the 
intuitive meaning given earlier, we mention that 
$\phi' \gc p \evaleqv{V} (\phi' \gc p) \seqc (\psi' \gc \ep)$
only if, for all $V$-evaluation maps $\sigma$, there exists a 
$V$-evaluation map $\sigma'$ such that
$\eval{\sigma}(\phi' \gc p) \bisim 
 \eval{\sigma}(\phi' \gc p) \seqc \eval{\sigma'}(\psi' \gc \ep)$.

Notice that, using the unary guard operator mentioned in 
Section~\ref{sect-deACPei}, we can write
$\guard{\phi'} \seqc p \evaleqv{V}
 \guard{\phi'} \seqc p \seqc \guard{\psi'}$
instead of
$\phi' \gc p \evaleqv{V} (\phi' \gc p) \seqc (\psi' \gc \ep)$.

Below, we will present the axioms and rules of \HL.
In addition to axioms and rules that concern a particular constant or 
operator of \mbox{\deACPei}, there is a rule concerning auxiliary 
flexible variables and a rule for precondition strengthening and/or 
postcondition weakening.

We use some special terminology and notations with respect to auxiliary
variables.
Let $p \in \HProcTerm$, and let $A \subseteq \FVar(p)$.\pagebreak[2]
Then $A$ is a \emph{set of auxiliary variables of} $p$ if each  
flexible variable in $A$ occurs in $p$ only in subterms of the form 
$\ass{v}{e}$ with $v \in A$.
We write $\AVars(p)$, where $p \in \HProcTerm$, for the set of all sets 
of auxiliary variables of $p$.
Moreover, we write $p_A$, where $p \in \HProcTerm$ and 
$A \in \AVars(p)$, for $p$ with all occurrences of subterms of the form 
$\ass{v}{e}$ with $v \in A$ replaced by $\ep$.

The axioms and rules of \HL\ are given in Table~\ref{hl-deACPei}.
\begin{table}[!t]
\caption{Axioms and rules of \HL}
\label{hl-deACPei}
\begin{druletbl}
{} \\[-3ex]
\text{inaction axiom:}
&
\HAxiom{\assproc{\phi}{\dead}{\psi}}
\\
\text{empty process axiom:}
&
\HAxiom{\assproc{\phi}{\ep}{\phi}}
\\
\text{basic action axiom:}
&
\HAxiom{\assproc{\phi}{a}{\phi}}
\\
\text{data parameterized action axiom:}
&
\HAxiom{\assproc{\phi}{a(e_1,\ldots,e_n)}{\phi}}
\\
\text{assignment axiom:}
&
\HAxiom{\assproc{\phi\subst{e}{v}}{\ass{v}{e}}{\phi}}
\\
\text{alternative composition rule:}
&
\HRule
{\assproc{\phi}{p}{\psi},\; \assproc{\phi}{q}{\psi}}
{\assproc{\phi}{p \altc q}{\psi}}
\\
\text{sequential composition rule:}
&
\HRule
{\assproc{\phi}{p}{\psi},\; \assproc{\psi}{q}{\chi}}
{\assproc{\phi}{p \seqc q}{\chi}}
\\
\text{iteration rule:}
&
\HRule
{\assproc{\phi}{p}{\phi},\; \assproc{\phi}{q}{\psi}}
{\assproc{\phi}{p \iter q}{\psi}}
\\
\text{guarded command rule:}
&
\HRule
{\assproc{\phi \Land \psi}{p}{\chi}}
{\assproc{\phi}{\psi \gc p}{\chi}}
\\
\text{parallel composition rule:}
&
\HRule
{\assproc{\phi}{p}{\psi},\; \assproc{\phi'}{q}{\psi'}}
{\assproc{\phi \Land \phi'}{p \parc q}{\psi \Land \psi'}}
\; \text{premises are disjoint}
\\
\text{encapsulation rule:}
&
\HRule
{\assproc{\phi}{p}{\psi}}
{\assproc{\phi}{\encap{H}(p)}{\psi}}
\\
\text{auxiliary variables rule:}
&
\HRule
{\assproc{\phi}{p}{\psi}}
{\assproc{\phi}{p_A}{\psi}}
\; A \in \AVars(p),\; \FVar(\psi) \inter A = \emptyset
\\
\text{consequence rule:}
&
\HRule
{\Der{}{\phi \Limpl \phi' = \True},\;
 \assproc{\phi'}{p}{\psi'},\;
 \Der{}{\psi' \Limpl \psi = \True}\;}
{\assproc{\phi}{p}{\psi}}
\\[-1.5ex]
\end{druletbl}
\end{table}
In this table, 
$p$ and~$q$ stand for arbitrary terms from $\HProcTerm$,
$\phi$, $\psi$, $\chi$, $\phi'$, and $\psi'$ stand for arbitrary terms 
from $\HCondTerm$,
$a$ stands for an arbitrary basic action from $\Act$,
$v$ stands for an arbitrary flexible variable from $\ProgVar$, and
$e$ and $e_1,e_2,\ldots$ stand for arbitrary terms from $\DataTerm$.
The parallel composition rule may only be applied if the premises are
disjoint.
Premises $\assproc{\phi}{p}{\psi}$ and 
$\assproc{\phi'}{q}{\psi'}$ are \emph{disjoint} if
\begin{itemize}
\item
$\FAss(p) \inter \FVar(q) = \emptyset$,
$\FAss(p) \inter \FVar(\phi') = \emptyset$, and
$\FAss(p) \inter \FVar(\psi') = \emptyset$;
\item
$\FAss(q) \inter \FVar(p) = \emptyset$,
$\FAss(q) \inter \FVar(\phi) = \emptyset$, \hspace*{.003em} and
$\FAss(q) \inter \FVar(\psi) = \emptyset$.
\end{itemize}
In the consequence rule, the first premise and the last premise are not 
asserted processes.
They assert that $\phi \Limpl \phi' = \True$ and 
$\psi' \Limpl \psi = \True$ are derivable from the axioms of \deACPei.

Before we move on to the soundness of the axioms and rules of \HL, we
consider two congruence related properties of the equivalences 
\smash{$\evaleqv{V}$} that are relevant to the soundness proof.
 
\begin{theorem}[Congruence]
\label{theorem-congruence-evaleqv}
For all finite $V \subseteq \ProgVar$, 
for all terms $p,q,p',q' \in \ProcTerm_V$, 
$p \evaleqv{V} p'$ and $q \evaleqv{V} q'$ only if 
$p \altc q \evaleqv{V} p' \altc q'$, 
$p \seqc q \evaleqv{V} p' \seqc q'$, and
$p \iter q \evaleqv{V} p' \iter q'$. 
Moreover, for all finite $V \subseteq \ProgVar$, 
for all terms $p,p' \in \ProcTerm_V$ 
and all terms $\phi \in \HCondTerm$ with $\FVar(\phi) \subseteq V$, 
$p \evaleqv{V} p'$ only if $\phi \gc p \evaleqv{V} \phi \gc p'$ and
$\encap{H}(p) \evaleqv{V} \encap{H}(p')$.
\end{theorem}
\begin{proof} 
Assume $p \evaleqv{V} p'$ and $q \evaleqv{V} q'$.
Then $p \altc q \evaleqv{V} p' \altc q'$ follows immediately and 
$p \seqc q \evaleqv{V} p' \seqc q'$ and 
$p \iter q \evaleqv{V} p' \iter q'$ 
follow easily by induction on the number of proper subprocesses of $p$, 
where use is made of Lemma~\ref{lemma-HNF}.
Assume $p \evaleqv{V} p'$.
Then $\phi \gc p \evaleqv{V} \phi \gc p'$ follows immediately and 
$\encap{H}(p) \evaleqv{V} \encap{H}(p')$ follows easily by induction on 
the number of proper subprocesses of $p$, where use is made of 
Lemma~\ref{lemma-HNF}.
\qed
\end{proof}

\begin{theorem}[Limited Congruence]
\label{theorem-congruence-evaleqv-parc}
\sloppy
For all finite $V \subseteq \ProgVar$,
for all terms $p,q,p',q' \in \ProcTerm_V$ with
$\FAss(p) \inter \FVar(q) = \emptyset$,
$\FAss(q) \inter \FVar(p) = \emptyset$,
$\FAss(p') \inter \FVar(q') = \emptyset$, and 
$\FAss(q') \inter \FVar(p') = \emptyset$,
$p \evaleqv{V} p'$ and $q \evaleqv{V} q'$ only if 
$p \parc q \evaleqv{V} p' \parc q'$. 
\end{theorem}
\begin{proof} 
\sloppy
Assume 
$\FAss(p) \inter \FVar(q) = \emptyset$ and
$\FAss(q) \inter \FVar(p) = \emptyset$, 
$\FAss(p') \inter\linebreak[2] \FVar(q') = \emptyset$ and
$\FAss(q') \inter \FVar(p') = \emptyset$, 
$p \evaleqv{V} p'$ and $q \evaleqv{V} q'$.
Then $p \parc q \evaleqv{V} p' \parc q'$ follows easily by induction on 
the number of proper subprocesses of $p$, where use is made of 
Lemma~\ref{lemma-HNF}.
\qed
\end{proof}

\begin{theorem}[Soundness]
\label{theorem-soundness-HL}
For all terms $p \in \HProcTerm$, 
for all terms $\phi,\psi \in \HCondTerm$, 
the asserted process $\assproc{\phi}{p}{\psi}$ is derivable from the 
axioms and rules of \HL\ only if $\assproc{\phi}{p}{\psi}$ is true.
\end{theorem}
\begin{proof} 
We will assume that $\phi,\psi \in \CondTerm$.
We can do so without loss of generality because, by the definition of
the truth of asserted processes, it is sufficient to consider arbitrary 
closed substitution instances of $\phi$ and $\psi$ if 
$\phi,\psi \notin \CondTerm$.
We will prove the theorem by proving that each of the axioms is true and
each of the rules is such that only true conclusions can be drawn from 
true premises.   
The theorem then follows by induction on the length of the proof.

The proofs for the axioms and the consequence rule are trivial.
Theorems~\ref{theorem-congruence-evaleqv}
and~\ref{theorem-congruence-evaleqv-parc} 
facilitate the proofs for the other rules.
By these theorems, the proofs for the alternative composition rule, the 
sequential composition rule, and the guarded command rule are also 
trivial and the proofs for the parallel composition rule, the 
encapsulation rule, and the auxiliary variables rule are straightforward 
proofs by induction on the number of proper subprocesses, in which use 
is made of Lemma~\ref{lemma-HNF}.
The parallel composition rule is proved simultaneously with similar 
rules for the left merge operator and the communication merge operator.
The proof for the iteration rule goes in a less straightforward way.

In case of the iteration rule, we assume that
\begin{itemize} 
\item[(1)]
for all $V$-evaluation maps $\sigma$,
$\eval{\sigma}(\phi \gc p) =
 \eval{\sigma}((\phi \gc p) \seqc (\phi \gc \ep))$ is derivable;
\item[(2)]
for all $V$-evaluation maps $\sigma$,
$\eval{\sigma}(\phi \gc q) =
 \eval{\sigma}((\phi \gc q) \seqc (\psi \gc \ep))$ is derivable;
\end{itemize}
and we prove that
\begin{itemize} 
\item[(3)]
for all $V$-evaluation maps $\sigma$,
$\eval{\sigma}(\phi \gc (p \iter q)) =
 \eval{\sigma}((\phi \gc (p \iter q)) \seqc (\psi \gc \ep))$ 
is derivable;
\end{itemize}
where $V = \FVar(\phi) \union \FVar(p \iter q) \union \FVar(\psi)$.
We do so by induction on the number of proper subprocesses of
$\eval{\sigma}(\phi \gc (p \iter q))$.

The basis step is trivial.
The inductive step is proved in the following way.
It follows easily from assumption~(1), making use of BKS1, that
\begin{itemize} 
\item[(4)]
for all $V$-evaluation maps $\sigma$,
for some evaluation map $\sigma'$,
$\eval{\sigma}(\phi \gc (p \iter q)) =
 \eval{\sigma}(\phi \gc p) \seqc
 \eval{\sigma'}(\phi \gc (p \iter q)) \altc
 \eval{\sigma}(\phi \gc q)$ 
is derivable.
\end{itemize}
We distinguish two cases: $\sigma \neq \sigma'$ and $\sigma = \sigma'$.

In the case where $\sigma \neq \sigma'$, (3) follows easily from~(4), 
the induction hypothesis, and assumption~(2), making use of BKS1.

In the case where $\sigma = \sigma'$, it follows immediately 
from~(4), making use of RSP*, that
\begin{itemize} 
\item[(5)]
for all $V$-evaluation maps $\sigma$,
$\eval{\sigma}(\phi \gc (p \iter q)) =
 \eval{\sigma}((\phi \gc p) \iter (\phi \gc q))$ 
is derivable;
\end{itemize}
and (3) follows easily from~(5) and assumption~(2), making use of BKS1.
\qed
\end{proof}
In the proof of Theorem~\ref{theorem-soundness-HL}, RSP* is used in the 
part concerning the iteration rule.
We do not see how that part of  the proof can be done if RSP* is 
replaced by BKS2--BKS5.

The following is a corollary of the definition of the truth of asserted 
processes and Theorem~\ref{theorem-soundness-HL}.
\begin{corollary}
\label{corollary-soundness}
For all terms $p, p' \in \HProcTerm$, 
for all terms $\phi,\psi \in \HCondTerm$, 
the asserted process $\assproc{\phi}{p}{\psi}$ is derivable from the 
axioms and rules of \HL\ and $p = p'$ is derivable from the axioms of
\deACPei\ only if $\assproc{\phi}{p'}{\psi}$ is true.
\end{corollary}

If it is possible at all, equational reasoning with the axioms of a 
process algebra about how data change in the course of a process is 
often rather cumbersome.
In many cases, but not all, reasoning with the axioms and rules of a 
Hoare logic is much more convenient.
We have not strived for a Hoare logic that covers the cases where it 
does not simplify reasoning.
Actually, the axioms and rules of \HL\ are not complete (in the sense of 
Cook~\cite{Coo78a}).
The side condition of the parallel composition rule precludes 
completeness.
We have, for example, that the asserted process
$\assproc{i = 0}{\ass{i}{i+1} \seqc \ass{i}{i+1} \parc \ass{i}{0}}
  {i = 0 \Lor i = 1 \Lor i = 2}$ 
is true, but this cannot be derived by means of the axioms and rules of
\HL\ alone because a premise of the form 
$\assproc{\phi}{\ass{i}{i+1} \seqc \ass{i}{i+1}}{\psi}$ and a premise of
the form $\assproc{\phi'}{\ass{i}{0}}{\psi'}$ are never disjoint.

We could have replaced the disjointness side condition by an 
interference-freedom side condition to cover cases such as the example 
given above and perhaps this would lead to completeness.
However, unless the disjointness side condition would suffice, 
fulfillment of the interference-freedom side condition generally needs a 
sophisticated proof.
These interference-freedom proofs partly outweigh the advantage of using 
a Hoare logic for reasoning about how data change in the course of a 
process.
As will be shown by means of an example in Section~\ref{sect-discussion}, 
equational reasoning with the axioms of \deACPei\ offers an alternative 
without interference-freedom proofs.
That is why we have chosen for the parallel composition rule with the 
disjointness side condition.

\section{On the Connection between the Hoare Logic and \deACPei}
\label{sect-using-HL}

In this section, we go into the connection of \HL\ with \deACPei\ by way
of the equivalence relations $\evaleqv{V}$.

Let $\assproc{\phi}{p}{\psi}$ be an asserted process, and 
let $V = \FVar(\phi) \union \FVar(p) \union \FVar(\psi)$.
Suppose that $\assproc{\phi}{p}{\psi}$ has been derived from the axioms 
and rules of \HL.
Then, by Theorem~\ref{theorem-soundness-HL}, $\assproc{\phi}{p}{\psi}$ 
is true.
This means that, for all closed substitution instances 
$\assproc{\phi'}{p}{\psi'}$ of $\assproc{\phi}{p}{\psi}$,
$\phi' \gc p \evaleqv{V} (\phi' \gc p) \seqc (\psi' \gc \ep)$.
In other words, for all closed substitution instances 
$\assproc{\phi'}{p}{\psi'}$ of $\assproc{\phi}{p}{\psi}$,
for all $V$-evaluation maps~$\sigma$, 
$\eval{\sigma}(\phi' \gc p) = 
 \eval{\sigma}(\phi' \gc p) \seqc \eval{\sigma'}(\psi' \gc \ep)$
is derivable from the axioms of \deACPei.
Thus, the derivation of $\assproc{\phi}{p}{\psi}$ from the axioms and 
rules of \HL\ has made a collection of equations available that can be 
considered to be derived by equational reasoning from the axioms of 
\deACPei.

Let us have a closer look at the equivalence relation $\evaleqv{V}$ on 
$\ProcTerm_V$.
Clearly, this equivalence relation is useful when reasoning about 
processes in which data are involved.
However, it is plain from the proof of 
Theorem~\ref{theorem-congruence-evaleqv-parc}
that $\evaleqv{V}$ is not a congruence relation on $\ProcTerm_V$.
This complicates the use of equational reasoning to derive, among other
things, the collection of equations referred to above considerably.
The presented Hoare logic can be considered to be a means to get 
partially round the complications concerned.

Dissociated from its connection with \HL, $\evaleqv{V}$ remains an 
interesting equivalence relation on $\ProcTerm_V$ when it comes to 
reasoning about processes in which data is involved.
Therefore, we mention below a result on this equivalence relation which 
is a corollary of results from Section~\ref{sect-HL} used to prove the 
soundness of \HL.
The fact that $\evaleqv{V}$ is not a congruence relation on 
$\ProcTerm_V$, and consequently that $\evaleqv{V}$ is not preserved by 
all contexts, makes this corollary to the point.
In order to formulate the corollary, we first define a set of contexts,
using $\Box$ as a placeholder.

For each finite $V \subseteq \ProgVar$, the set $\sContext{V}$ of 
\emph{sequential evaluation supporting contexts for $V$} is the set 
$\Union_{W \subseteq V} \sContext{V,W}$, where the sets 
$\sContext{V,W}$, for finite $V,W \subseteq \ProgVar$ with 
$W \subseteq V$, are defined by simultaneous induction as follows:
\begin{itemize}
\item
$\Box \in \sContext{V,W}$;
\item
if $p \in \sProcTerm$, $C \in \sContext{V,W}$, 
$\FVar(p) \subseteq V$, and $\FAss(p) \subseteq W$, then 
$p \altc C,\; C \altc p,\linebreak[2]
 p \seqc C,\; C \seqc p,\; p \iter C,\; C \iter p \in \sContext{V,W}$;
\item
if $\phi \in \CondTerm$ and $C \in \sContext{V,W}$,
$\FVar(\phi) \subseteq V$, then $\phi \gc C \in \sContext{V,W}$;
\item
if $p \in \sProcTerm$, $C \in \sContext{V,W}$, 
$\FAss(p) \inter V = \emptyset$, and $\FVar(p) \inter W = \emptyset$, 
then 
$p \parc C,\;
 C \parc p \in \sContext{V \union \FVar(p),W \union \FAss(p)}$; 
\vspace*{0.25ex}
\item
if $H \subseteq \Act$ and $C \in \sContext{V,W}$, then 
$\encap{H}(C) \in \sContext{V,W}$.
\end{itemize}
We write $C[p]$, where $C \in \sContext{V}$ and $p \in \ProcTerm$, for
$C$ with the occurrence of $\Box$ replaced by $p$.

The following is a corollary of 
Theorems~\ref{theorem-congruence-evaleqv}
and~\ref{theorem-congruence-evaleqv-parc}.
\begin{corollary}
\label{corollary-contexts}
Let $V$ be a finite subset of $\ProgVar$.
Then, for all $p,p' \in \ProcTerm_V$, for all $C \in \sContext{V}$,
$p \evaleqv{V} p'$ only if $C[p] \evaleqv{V} C[p']$.
\end{corollary}
Of course, Corollary~\ref{corollary-contexts} can be applied to results
from using \HL.
Let $\assproc{\phi}{p}{\psi}$ be an asserted process,  
let $V = \FVar(\phi) \union \FVar(p) \union \FVar(\psi)$, and 
let $C \in \sContext{V}$.
Suppose that $\assproc{\phi}{p}{\psi}$ has been derived from the axioms 
and rules of \HL.
Then, for all closed substitution instances 
$\assproc{\phi'}{p}{\psi'}$ of $\assproc{\phi}{p}{\psi}$, we have that
$C[\phi' \gc p] \evaleqv{V} C[(\phi' \gc p) \seqc (\psi' \gc \ep)]$.

\section{On the Role of the Hoare Logic for \deACPei}
\label{sect-discussion}

Process algebras focus on the main role of a reactive system, namely 
maintaining some ongoing interaction with its environment.
Hoare logics focus on the main role of a transformational system, namely 
producing, without interruption by its environment, outputs from 
inputs.%
\footnote
{The terms reactive system and transformational system were coined
 in~\cite{HP85a}.}
However, actual systems are often reactive systems composed of reactive 
components and transformational components.
\deACPei\ provides a setting for equational reasoning about the 
behaviour of such systems, but it does not offer by itself the 
possibility to reason in Hoare-logic style about the behaviour of the 
transformational components.

Below, we will take the behaviour of a very simple transformational 
component and reason about how it changes data both in Hoare-logic style 
with the axioms and rules of \HL\ and equationally with the axioms of
\deACPei.
We assume that $\gD$ is the group of integers.
We also assume that $i$ and~$j$ are flexible variables from 
$\ProgVar$ and $n$ and $n'$ are variables of sort $\Data$.
Moreover, we use $e - e'$ as an abbreviation of $e + (-e')$.
The behaviour of the very simple transformational component concerned is 
described by the closed \deACPei\ term
$\ass{i}{i + j} \seqc \ass{j}{i - j} \seqc \ass{i}{i - j}$. 

We begin with showing by means of \HL\ that this behaviour swaps the 
values of $i$ and $j$.
We derive
\begin{ldispl}
\assproc{i = n \Land j = n'}{\ass{i}{i + j}}{i = n + n' \Land j = n'}
\end{ldispl}%
using the assignment axiom and the consequence rule. 
Similarly, we derive
\begin{ldispl}
 \assproc{i = n + n' \Land j = n'}{\ass{j}{i - j}}
  {i = n + n' \Land j = n}
\end{ldispl}%
and
\begin{ldispl}
\assproc{i = n + n' \Land j = n}{\ass{i}{i - j}}{i = n' \Land j = n}\;.
\end{ldispl}%
From these three asserted processes, we derive
\begin{ldispl}
 \assproc{i = n \Land j = n'}
  {\ass{i}{i + j} \seqc \ass{j}{i - j} \seqc \ass{i}{i - j}}
  {i = n' \Land j = n}
\end{ldispl}%
using the sequential composition rule twice.

We continue with showing the same by means of \deACPei.
This means that we have to derive from the axioms of \deACPei, for all
$e,e' \in \DataVal$, for all $\set{i,j}$-evaluation maps $\sigma$:
\begin{ldispl}
(*) \quad
\begin{array}[c]{@{}l@{}}
 \eval{\sigma}
  ((i = e \Land j = e') \gc
   \ass{i}{i + j} \seqc \ass{j}{i - j} \seqc \ass{i}{i - j}) \\ 
\; {} =
 \eval{\sigma}
  ((i = e \Land j = e') \\
\phantom{\; {} = \eval{\sigma}(}
\; {} \gc
   \ass{i}{i + j} \seqc \ass{j}{i - j} \seqc \ass{i}{i - j} \seqc
   (i = e' \Land j = e) \gc \ep)\;.
\end{array}
\end{ldispl}%
We derive 
\begin{ldispl}
 \eval{\sigma}((i = e \Land j = e') \gc \ass{i}{i + j}) = 
 \sigma(i = e \Land j = e') \gc \ass{i}{\sigma(i + j)}
\end{ldispl}%
using axioms V3 and V5; and 
\begin{ldispl}
 \eval{\sigma}
  ((i = e \Land j = e') \gc \ass{i}{i + j} \seqc
   (i = e + e' \Land j = e') \gc \ep) \\
\; {} =
 \sigma(i = e \Land j = e') \\
\phantom{\; {} = {}}
\; {} \gc
 \ass{i}{\sigma(i + j)} \seqc
 \sigma\mapupd{\sigma(e + e')}{i}(i = e + e' \Land j = e') \gc \ep
\end{ldispl}%
using axioms V0, V3, and V5.
\pagebreak[2]
\\
We can derive the following equation for all $\set{i,j}$-evaluation maps 
$\sigma$:
\begin{ldispl}
(**) \quad
\begin{array}[c]{{@{}l@{}}}
 \sigma(i = e \Land j = e') \gc \ass{i}{\sigma(i + j)} \\
\; {} =
 \sigma(i = e \Land j = e') \\
\phantom{\; {} = {}}
\; {} \gc
 \ass{i}{\sigma(i + j)} \seqc
 \sigma\mapupd{\sigma(e + e')}{i}(i = e + e' \Land j = e') \gc \ep\;. 
\end{array}
\end{ldispl}%
In the case $\sigma(i) = e$ and $\sigma(j) = e'$, we derive
$\sigma\mapupd{\sigma(e + e')}{i}(i = e + e' \Land j = e') = \True$
using IMP2.
From this, we derive equation~(**) using axioms GC1 and A8.
\\
In the case $\sigma(i) \neq e$ or $\sigma(j) \neq e'$, we derive
$\sigma\mapupd{\sigma(e + e')}{i}(i = e + e' \Land j = e') = \False$
using IMP2.
From this, we derive equation~(**) using axiom GC2.
\\
Hence, we have for all $\set{i,j}$-evaluation maps $\sigma$:
\begin{ldispl}
 \eval{\sigma}((i = e \Land j = e') \gc \ass{i}{i + j}) \\
\; {} = 
 \eval{\sigma}
  ((i = e \Land j = e') \gc \ass{i}{i + j} \seqc
   (i = e + e' \Land j = e') \gc \ep)\;.
\end{ldispl}%
Similarly, we find  for all $\set{i,j}$-evaluation maps $\sigma$:
\begin{ldispl}
 \eval{\sigma}((i = e + e' \Land j = e') \gc \ass{j}{i - j}) \\
\; {} = 
 \eval{\sigma}
  ((i = e + e' \Land j = e') \gc \ass{j}{i - j} \seqc
   (i = e + e' \Land j = e) \gc \ep)
\end{ldispl}%
and
\begin{ldispl}
 \eval{\sigma}((i = e + e' \Land j = e) \gc \ass{i}{i - j}) \\
\; {} = 
 \eval{\sigma}
  ((i = e + e' \Land j = e) \gc \ass{i}{i - j} \seqc
   (i = e' \Land j = e) \gc \ep)\;.
\end{ldispl}%
From the last three equations, we derive equation~(*) using axioms A5, 
A9, GC5, V3, and V5.
By this we have finally shown by means of \deACPei\ that the values of 
$i$ and $j$ are swapped by the process described by
$\ass{i}{i + j} \seqc \ass{j}{i - j} \seqc \ass{i}{i - j}$.

In this case, it is clear that Hoare-logic style reasoning with the 
axioms and rules of \HL\ is much more convenient than equational 
reasoning with the axioms of \deACPei.
Because a single application of a rule of \HL\ cannot be justified by a 
single application of an axiom of \deACPei, we expect that this also 
holds for virtually all other cases of reasoning about how the behaviour 
of a transformational system changes data.

Now, we turn our attention to the rather restrictive side condition of 
the parallel composition rule of \HL.
As mentioned before at the end of Section~\ref{sect-HL}, we have that
the asserted process
\begin{ldispl}
 \assproc{i = 0}{\ass{i}{i+1} \seqc \ass{i}{i+1} \parc \ass{i}{0}}
  {i = 0 \Lor i = 1 \Lor i = 2} 
\end{ldispl}%
is true, but this cannot be derived by means of the axioms and rules of
\HL\ alone because a premise of the form 
$\assproc{\phi}{\ass{i}{i+1} \seqc \ass{i}{i+1}}{\psi}$ and a premise of
the form $\assproc{\phi'}{\ass{i}{0}}{\psi'}$ are never disjoint.
However, we can derive the following equation from the axioms of 
\deACPei: 
\begin{ldispl}
 \ass{i}{i+1} \seqc \ass{i}{i+1} \parc \ass{i}{0} \\ 
\; {} =
 \ass{i}{i+1} \seqc 
 (\ass{i}{i+1} \seqc \ass{i}{0} \altc \ass{i}{0} \seqc \ass{i}{i+1}) \\
\phantom{\; {} = {}}
\; {} \altc
 \ass{i}{0} \seqc \ass{i}{i+1} \seqc \ass{i}{i+1}\;.
\end{ldispl}%
By Corollary~\ref{corollary-soundness}, it is sound to replace in the 
above asserted process the left-hand side of this equation by the 
right-hand side of this equation. 
This yields the asserted process 
\begin{ldispl}
\{i = 0\} \\
\{\ass{i}{i+1} \seqc 
  (\ass{i}{i+1} \seqc \ass{i}{0} \altc \ass{i}{0} \seqc \ass{i}{i+1}) \\
\; \; {} \altc
  \ass{i}{0} \seqc \ass{i}{i+1} \seqc \ass{i}{i+1}\} \\
\{i = 0 \Lor i = 1 \Lor i = 2\}\;,
\end{ldispl}%
which can be derived using the assignment axiom, the alternative 
composition rule, the sequential composition rule, and the consequence 
rule of \HL\ several times.

If the disjointness side condition of the parallel composition rule of 
\HL\ is replaced by an interference-freedom side condition, like
in~\cite{OG76a}, then the original asserted process becomes derivable 
using the axioms and rules of the Hoare logic alone 
(see e.g.~\cite[page~278]{ABO09a}).
The interference-freedom proof involved needs proof outlines
(see~\cite{OG76a}) for 
$\assproc{i = 0}{\ass{i}{i+1} \seqc \ass{i}{i+1}}{\True}$ and
$\assproc{\True}{\ass{i}{0}}{i = 0 \Lor i = 1 \Lor i = 2}$. 
In this very simple case, the interference-freedom proof already amounts 
to seven interference-freedom checks.
However, for two processes in which $k$ and $k'$ assignment actions 
occur, the number of interference-freedom checks is at least
$2 \mul k \mul k' + k + k'$.
Therefore, we expect that interference-freedom proofs partly outweigh 
the advantage of using a Hoare logic.

\section{Related Work}
\label{sect-related}

The approach to the formal verification of programs that is now known as
Hoare logic was proposed in~\cite{Hoa69a}.
The illustration of this approach was at the time confined to the very 
simple deterministic sequential programs that are mostly referred to as 
while programs (cf.~\cite{ABO09a}).
The axioms, the sequential composition rule, the iteration rule, the
guarded command rule, and the consequence rule from our Hoare logic 
savour strongly of the common rules for while programs.
The alternative composition rule is the or rule due to~\cite{Lau71a}, 
the parallel composition rule was proposed in~\cite{Hoa72a}, and the 
auxiliary variables rule was first introduced in~\cite{OG76a}.   
The parallel composition rules proposed in~\cite{AFR80a,LG81a,OG76a} are
more complicated than our parallel composition rule.

In the case of~\cite{AFR80a,LG81a}, the intention was to provide a Hoare 
logic for the first design of CSP~\cite{Hoa78}.
In that design, one program may force another program to assign a data 
value sent by the former program to a program variable used by the 
latter program.
This feature complicates the parallel composition rule considerably.
Moreover, incorporating this feature in an ACP-like process algebra 
would lead to the situation that, in equational reasoning, certain 
axioms may not be applied in contexts of parallel processes (like
in~\cite{GP94b}, see below).
Because our concern is in the use of a Hoare logic as a complement to 
pure equational reasoning, we have not considered incorporating this 
feature.
 
In the case of~\cite{OG76a}, the rule is more complicated because, in 
the parallel programs covered, program variables may be shared 
variables, i.e.\ program variables that are assigned to in one program 
may be used in another program.
Our process algebra also covers shared variables.
However, covering shared variables in our Hoare logic as well would mean 
that the simple disjointness proof required by our parallel composition 
rule has to be replaced a sophisticated interference-freedom proof.
We believe that this would diminish the usefulness of our Hoare logic as
a complement to equational reasoning considerably.
Therefore, we have not considered covering shared variables in the
parallel composition rule.

In~\cite{GP94b}, an extension of ACP with the empty process constant and 
the unary counterpart of the binary guarded command operator is 
presented, the truth of an asserted sequential process is defined in 
terms of the transition relations from the given structural operational 
semantics of the presented extension of ACP, and it is shown that an 
asserted sequential process $\assproc{\phi}{p}{\psi}$ is true according 
to that definition iff 
$\guard{\phi} \seqc p \bisim' \guard{\phi} \seqc p \seqc \guard{\psi}$, 
where $\bisim'$ is bisimulation equivalence as defined in~\cite{GP94b} 
for sequential processes.
Moreover, a Hoare logic of sequential asserted processes is presented 
and its soundness is shown. 
However, \cite{GP94b} does not go into the use of that Hoare logic as a
complement to pure equational reasoning from the equational axioms. 

Regarding the bisimulation equivalence $\bisim'$ defined in~\cite{GP94b} 
for sequential processes, we can mention that, if the data-states are
evaluation maps, $p \bisim' q$ iff 
$\eval{\sigma}(p) \bisim \eval{\sigma}(q)$ for all $V$-evaluation maps 
$\sigma$, where $V = \FVar(p) \union \FVar(q)$.
Due to the possibility of interference between parallel processes, a 
different bisimulation equivalence $\bisim''$, finer than $\bisim'$, is 
needed in~\cite{GP94b} for parallel processes.
As a consequence, in equational reasoning, certain axioms may not be 
applied in contexts of parallel processes.
Moreover, $\bisim$ together with the operators $\eval{\sigma}$ allows of 
dealing with local data-states, whereas the combination of $\bisim'$ and 
$\bisim''$ does not allow of dealing with local data-states. 

\section{Concluding Remarks}
\label{sect-conclusions}

We have taken an extension of \ACP\ with features that are relevant to 
processes in which data are involved, devised a Hoare logic of asserted
processes based on this extension of \ACP, and gone into the use of this 
Hoare logic as a complement to pure equational reasoning from the axioms 
of the extension of \ACP.

We have defined what it means that an asserted process is true in terms 
of an equivalence relation ($\evaleqv{V}$) that had been found to be 
central to relating the extension of \ACP\ and the Hoare logic. 
That this equivalence relation is not a congruence relation with respect 
to parallel composition is related to the fact that in the extension of 
\ACP\ presented in~\cite{GP94b} certain axioms may not be applied in 
contexts of parallel processes.

In this paper, we build on earlier work on ACP. 
The axioms of \ACPe\ have been taken from~\cite[Section 4.4]{BW90}, the 
axioms for the iteration operator have been taken from~\cite{BFP01a}, 
and the axioms for the guarded command operator have been taken 
from~\cite{BB92c}.
The  evaluation operators have been inspired by~\cite{BM05a} and the data 
parameterized action operator has been inspired by~\cite{BM09d}.

\bibliographystyle{splncs03}
\bibliography{PA}

\end{document}